%% file: derivatives2.tex
\RequirePackage{fix-cm}
\RequirePackage{fixltx2e}
\documentclass[twocolumn,english,authoryear,10pt]{sigplanconf}
\pdfoutput=1

\usepackage[T1]{fontenc}
\usepackage{babel}
\usepackage{prettyref}
\usepackage{mathtools}
\usepackage{amsmath}
\usepackage{amsthm}
\usepackage{stmaryrd}
\usepackage[authoryear]{natbib}
\usepackage[unicode=true,pdfusetitle,
 bookmarks=true,bookmarksnumbered=false,bookmarksopen=false,
 breaklinks=false,pdfborder={0 0 1},backref=false,colorlinks=false]
 {hyperref}

\makeatletter
\usepackage{fixsigplanconf}
\usepackage{calc} % For time calculations

\titlebanner{PREPRINT - DO NOT DISTRIBUTE}
\preprintfooter{PREPRINT - DO NOT DISTRIBUTE}

\makeatletter
\newcommand{\lyxkeywords}[1]{\hypersetup{pdfkeywords={#1}}\keywords{#1}}
\newcommand{\lyxdoi}[1]{\newcommand\@doi{#1}}
\newcommand{\lyxpublication}[1]{\newcommand\@publication{#1}}
\newcommand{\lyxissue}[1]{\newcommand\@issue{#1}}

\newcounter{hours}
\newcounter{minutes}
\newcommand\twodigits[1]{\ifnum#1<10 0\number#1 \else #1\fi}

\if \@preprint
  \def \ps@plain {%
    \let \@mkboth = \@gobbletwo
    \let \@evenhead = \@empty
    \def \@evenfoot {\scriptsize
                   \rlap{\textit{\@preprintfooter}}\hfil
                   \thepage \hfil
                   \llap{\textit{1\the\year%
                         \the\month.\twodigits{\the\day}%
                         \thehours\twodigits{\theminutes}}}}%
    \let \@oddhead = \@empty
    \let \@oddfoot = \@evenfoot}
\fi

\makeatother
\theoremstyle{plain}
\newtheorem{thm}{\protect\theoremname}
  \theoremstyle{plain}
  \newtheorem{lem}[thm]{\protect\lemmaname}
  \theoremstyle{definition}
  \newtheorem{defn}[thm]{\protect\definitionname}
 \newenvironment{lyxcode}
         {\par\begin{list}{}{
                 \setlength{\listparindent}{0pt}% needed for AMS classes
                 \raggedright
                 \setlength{\itemsep}{0pt}
                 \setlength{\parsep}{0pt}
                 \normalfont\ttfamily}%
          \item[]}
         {\end{list}}
\newcommand{\code}[1]{\texttt{#1}}

\usepackage{balance}
\usepackage{mathdots}

\newrefformat{fig}{Figure~\ref{#1}}
\newrefformat{sub}{Section~\ref{#1}}
\newrefformat{alg}{Algorithm~\ref{#1}}
\newrefformat{ctn}{Criterion~\ref{#1}}
\newrefformat{def}{Definition~\ref{#1}}

\usepackage{etex} % Fix "No room for new \dimen" error
\usepackage{tikz}
\usetikzlibrary{external,matrix,arrows,positioning,chains,scopes,fit,shapes,decorations,decorations.pathreplacing,calc,trees,patterns}
\usepackage{pgfplots}
\pgfplotsset{compat=1.3}
\usepackage{tikz-qtree}

\usepackage{bussproofs}
\usepackage{amsmath}
\usepackage{stmaryrd}

\usepackage[caption=false]{subfig}
\makeatletter

\def \ps@plain {%
  \let \@mkboth = \@gobbletwo
  \let \@evenhead = \@empty
  \def \@evenfoot {\scriptsize\hfil\thepage\hfil}%
  \let \@oddhead = \@empty
  \let \@oddfoot = \@evenfoot}

\@ifundefined{showcaptionsetup}{}{%
 \PassOptionsToPackage{caption=false}{subfig}}
\usepackage{subfig}
\AtBeginDocument{
  
}

\makeatother

  \providecommand{\definitionname}{Definition}
  \providecommand{\lemmaname}{Lemma}
\providecommand{\theoremname}{Theorem}

\begin{document}
%\toappear{}

\include{util/lyxproofs}

\titlebanner{PREPRINT}
\preprintfooter{PREPRINT}

\let\funarrowX=\rightarrow

\global\long\def\funarrow{\funarrowX}

\renewcommand\rightarrow{\ensuremath{\lhook\joinrel\funarrowX}}

%\useauthorversion
%\setcopyright{acmlicensed}
%\acmPrice{\$15.00}
%\conferenceinfo{
%\CopyrightYear{2016}
%\setcopyright{acmlicensed}
\publicationrights{licensed}
\conferenceinfo{PLDI '16, Proceedings of the 37th ACM SIGPLAN Conference on Programming Language Design and Implementation}{June 13 -- 17, 2016, Santa Barbara, CA, USA}
\copyrightyear{2016}
\copyrightdata{978-1-4503-4261-2/16/06}
\reprintprice{\$15.00}
\copyrightdoi{2908080.2908128}

\title{}

\title{On the Complexity and Performance of Parsing with Derivatives}

\authorinfo{Michael D. Adams}{University of Utah, USA}{}

\authorinfo{Celeste Hollenbeck}{University of Utah, USA}{}

\authorinfo{Matthew Might}{University of Utah, USA}{}
\maketitle
\begin{abstract}
Current algorithms for context-free parsing inflict a trade-off between
ease of understanding, ease of implementation, theoretical complexity,
and practical performance. No algorithm achieves all of these properties
simultaneously.

\citet{Might:2011:PDF:2034773.2034801} introduced parsing with derivatives,
which handles arbitrary context-free grammars while being both easy
to understand and simple to implement. Despite much initial enthusiasm
and a multitude of independent implementations, its worst-case complexity
has never been proven to be better than exponential. In fact, high-level
arguments claiming it is fundamentally exponential have been advanced
and even accepted as part of the folklore. Performance ended up being
sluggish in practice, and this sluggishness was taken as informal
evidence of exponentiality.

In this paper, we reexamine the performance of parsing with derivatives.
We have discovered that it is not exponential but, in fact, cubic.
Moreover, simple (though perhaps not obvious) modifications to the
implementation by \citet{Might:2011:PDF:2034773.2034801} lead to
an implementation that is not only easy to understand but also highly
performant in practice.

\marginpar{} 
\end{abstract}
%\category{D.3.1}{Programming Languages}{Formal Definitions and Theory}[Semantics]

\lyxkeywords{Parsing; Parsing with derivatives; Performance}

\section{Introduction}

Although many programmers have some familiarity with parsing, few
understand the intricacies of how parsing actually works. Rather than
hand-write a parser, many choose to use an existing parsing tool.
However, these tools are known for their maintenance and extension
challenges, vague error descriptions, and frustrating shift/reduce
and reduce/reduce conflicts \citep{parsingNon-LRwYacc}. 

In a bid to improve accessibility, \citet{Might:2011:PDF:2034773.2034801}
present a simple technique for parsing called parsing with derivatives
(PWD). Their parser extends the \citeauthor{bzd} derivative of regular
expressions \citep{bzd} to support context-free grammars (CFGs).
It transparently handles language ambiguity and recursion and is easy
to implement and understand.

PWD has been implemented in a number of languages \citep{McGuire2012,Vognsen2012,Mull2013,Shearar2013,Byrd2013,Engelberg2015,Pfiel2015}.
However, these tend to perform poorly, and many conjectured that the
algorithm is fundamentally exponential \citep{Cox:2010:Online,spiewak-email}
and could not be implemented efficiently. In fact, \citet{Might:2011:PDF:2034773.2034801}
report that their implementation took two seconds to parse only 31
lines of Python.

In this paper, we revisit the complexity and performance of PWD. It
turns out that the run time of PWD is linearly bounded by the number
of grammar nodes constructed during parsing, and we can strategically
assign \emph{unique names} to these nodes in such a way that the number
of possible names is \emph{cubic}. This means that the run time of
of PWD is, in fact, cubic, and the assumed exponential complexity
was illusory.

Investigating further, we revisit the implementation of PWD by \citet{Might:2011:PDF:2034773.2034801}
by building and carefully profiling a new implementation to determine
bottlenecks adversely affecting performance. We make three significant
improvements over the original algorithm: \emph{accelerated fixed
points}, \emph{improved compaction}, and \emph{more efficient memoization}.
Once these are fixed, PWD's performance improves to match that of
other general CFG parsers.

This paper makes the following contributions:
\begin{itemize}
\item \prettyref{sec:Background} reviews the work by \citet{Might:2011:PDF:2034773.2034801}
on PWD and its key ideas.
\item \prettyref{sec:Complexity-Analysis} investigates PWD's complexity
and shows that its upper bound is cubic instead of the exponential
that was previously believed. This makes PWD's asymptotic behavior
on par with that of other general CFG parsers.
\item \prettyref{sec:Practical-Performance} examines PWD's performance
and shows that targeted algorithmic improvements can achieve a speedup
of almost 1000 times over the implementation in \citet{Might:2011:PDF:2034773.2034801}.\newpage{}
\end{itemize}

\section{Background\label{sec:Background}}

\citet{bzd} presents and \citet{Owens:2009:RDR:1520284.1520288}
expand upon derivatives of regular expressions as a means to recognize
strings that match a given regular expression.

With PWD, \citet{Might:2011:PDF:2034773.2034801} extend the concept
of string recognition via \citeauthor{bzd} derivatives to CFGs. The
essential trick to this is handling recursively defined languages.
Computing the derivative of a non-terminal may require the derivative
of that same non-terminal again, causing an infinite loop. \citet{Might:2011:PDF:2034773.2034801}
circumvent this using a combination of memoization, laziness, and
fixed points. We briefly review their technique in this section.

\subsection{The \citeauthor{bzd} Derivative\label{sub:The-Brzozowski-derivative}}

\citet{bzd} matches regular expressions against an input by successively
matching each character of the input against the set of words in the
semantics of that regular expression. In three steps, he computes
the set of words (if any) that can validly appear after the initial
input character. First, he takes the first character of the input
and compares it to the first characters of the words in the semantics.
Second, he keeps only the words whose first characters match and discards
all others. Finally, he removes the first character from the remaining
words. 

\citet{bzd} calls this the derivative of a language and formally
defines it as the following, where $c$ is the input character and
$\llbracket L\rrbracket$ is the set of words in the language $L$:\textbf{
\[
D_{c}\left(L\right)=\left\{ w\mid cw\in\llbracket L\rrbracket\right\} 
\]
}For example, with respect to the character \code{f}, the derivative
of the language for which $\llbracket L\rrbracket=\{\code{foo},\code{frak},\code{bar}\}$
is $D_{\code{f}}\left(L\right)=\{\code{oo},\code{rak}\}$. Because
\code{foo} and \code{frak} start with the character \code{f} and
\code{bar} does not, we keep only \code{foo} and \code{frak} and
then remove their initial characters, leaving \code{oo} and \code{rak}.

We repeat this process with each character in the input until it is
exhausted. If, after every derivative has been performed, the resulting
set of words contains the empty word, $\epsilon$, then there is some
word in the original language consisting of exactly the input characters,
and the language accepts the input. All f this processing takes place
at parse time, so there is no parser-generation phase.

\subsection{Parsing Expressions\label{sub:Parsing-expressions}}

\begin{figure}[tb]
\textbf{Forms}
\begin{alignat*}{1}
L & ::=\emptyset\mid\epsilon_{s}\mid c\mid L_{1}\circ L_{2}\mid L_{1}\cup L_{2}\mid L\rightarrow f\\
s,t & \in T\mbox{\qquad\qquad Abstract syntax trees}\\
f & \in T\funarrow T\mbox{\,\,\,\,\,\,\,\,\,\,Reduction functions}
\end{alignat*}

\textbf{Semantics}
\[
\begin{alignedat}{2}\llbracket L\rrbracket & \in\wp\left(\Sigma^{*}\times T\right)\\
\llbracket\emptyset\rrbracket & =\left\{ \right\}  &  & \mbox{Empty Lang.}\\
\llbracket\epsilon_{s}\rrbracket & =\left\{ \left(\epsilon,s\right)\right\}  &  & \mbox{Empty Word}\\
\llbracket c\rrbracket & =\left\{ \left(c,c\right)\right\}  &  & \mbox{Token}\\
\llbracket L_{1}\circ L_{2}\rrbracket & =\{\left(uv,\left(s,t\right)\right)\mid\left(u,s\right)\in\llbracket L_{1}\rrbracket &  & \mbox{Concatenation}\\
 & \hphantom{{}=\{\left(uv,\left(s,t\right)\right)\mid{}}\llap{\mbox{and }}\left(v,t\right)\in\llbracket L_{2}\rrbracket\}\\
\llbracket L_{1}\cup L_{2}\rrbracket & =\{\left(u,s\right)\mid\left(u,s\right)\in\llbracket L_{1}\rrbracket &  & \mbox{Alternation}\\
 & \hphantom{{}=\{\left(u,s\right)\mid{}}\llap{\mbox{or }}\left(u,s\right)\in\llbracket L_{2}\rrbracket\}\\
\llbracket L\rightarrow f\rrbracket & =\left\{ \left(w,f\,s\right)\mid\left(w,s\right)\in\llbracket L\rrbracket\right\}  &  & \mbox{Reduction}
\end{alignedat}
\]

\centering{}\caption{\label{fig:parsing-expression-forms}Parsing expression forms}
\end{figure}

Explicitly enumerating the possibly infinite set of words in a language
can be cumbersome, so we express regular languages using the expression
forms in \prettyref{fig:parsing-expression-forms}. For the most part,
these consist of the traditional regular expression forms. The $\epsilon_{s}$
form is the language of the empty string, $\emptyset$ is the empty
language, $c$ is a single token, $(\circ)$ concatenates, and $(\cup)$
forms alternatives. In \citet{Might:2011:PDF:2034773.2034801}, every
expression also produces an abstract syntax tree (AST) upon success.
So, $\epsilon_{s}$ is annotated with a subscript $s$ indicating
the AST to be returned, and the reduction form $L\rightarrow f$ behaves
like $L$, except that it returns the result of applying $f$ to the
AST returned by $L$. The semantics of these forms are as in \prettyref{fig:parsing-expression-forms}
and are defined as sets of accepted strings paired with the AST that
returns for that string. For the purposes of parsing single tokens,
$c$, and concatenations, $\left(\circ\right)$, we assume the type
of ASTs includes tokens and pairs of ASTs.

Note that in this paper, we use $\epsilon$ for the empty word and
$\epsilon_{s}$ for the parsing expression that represents a language
containing only the empty word. Similarly, we use $c$ to refer to
either the single-token word or the parsing expression signifying
a language containing only one token.

Also, although \citet{Might:2011:PDF:2034773.2034801} include a form
for Kleene star, we omit this. Once these forms are extended from
regular expressions to CFGs in \prettyref{sub:Derivatives-of-context-free},
any use of Kleene star can be replaced with a definition like the
following.
\[
L^{*}=\epsilon_{s}\cup\left(L\circ L^{*}\right)
\]

\subsection{Derivatives of Parsing Expressions}

The derivatives of the language forms in \prettyref{fig:parsing-expression-forms}
with respect to a token $c$ are shown in \prettyref{fig:Derivatives-of-parsing-expression-forms}.
The derivative of $\emptyset$ is $\emptyset$, as $\llbracket\emptyset\rrbracket$
contains no words beginning with any character. For the same reason,
the derivative of $\epsilon_{s}$ is also $\emptyset$. The derivative
of a token $c$ depends on whether the input token matches $c$; the
result is $\epsilon_{c}$ if the input token matches and $\emptyset$
if not. The derivatives of $L_{1}\cup L_{2}$ and $L\rightarrow f$
merely take the derivatives of their children.

The derivative of $L_{1}\circ L_{2}$ has two cases, depending on
whether $\llbracket L_{1}\rrbracket$ contains $\epsilon$. If $\llbracket L_{1}\rrbracket$
\emph{does not} contain $\epsilon$, every word in the concatenated
language starts with a non-empty word from $L_{1}$. This means the
derivative of $L_{1}\circ L_{2}$ filters and removes the first token
from the words in $L_{1}$ while leaving $L_{2}$ alone. Thus, the
derivative of $L_{1}\circ L_{2}$ if $L_{1}$ does not contain $\epsilon$
is $D_{c}\left(L_{1}\right)\circ L_{2}$.

On the other hand, if $\llbracket L_{1}\rrbracket$ \emph{does} contain
$\epsilon$, then the derivative contains not only all the words in
$D_{c}\left(L_{1}\right)\circ L_{2}$ but also derivatives for when
the $\epsilon$ in $L_{1}$ is concatenated with words in $L_{2}$.
Since these concatenations are all words from $L_{2}$, this adds
$D_{c}\left(L_{2}\right)$ to the derivative. In this case, $D_{c}(L_{1}\circ L_{2})$
is therefore $(D_{c}\left(L_{1}\right)\circ L_{2})\cup D_{c}\left(L_{2}\right)$.

\subsection{\label{sub:Nullability}Nullability}

\begin{figure}[tb]
\[
\begin{alignedat}{1}D_{c}\left(\emptyset\right) & =\emptyset\\
D_{c}\left(\epsilon\right) & =\emptyset\\
D_{c}\left(c'\right) & =\begin{cases}
\epsilon_{c} & \mbox{ if }c=c'\\
\emptyset & \mbox{ if }c\ne c'
\end{cases}\\
D_{c}\left(L_{1}\cup L_{2}\right) & =D_{c}\left(L_{1}\right)\cup D_{c}\left(L_{2}\right)\\
D_{c}\left(L_{1}\circ L_{2}\right) & =\begin{cases}
\phantom{(}D_{c}\left(L_{1}\right)\circ L_{2} & \mbox{ if }\epsilon\notin\llbracket L_{1}\rrbracket\\
(D_{c}\left(L_{1}\right)\circ L_{2})\cup D_{c}\left(L_{2}\right) & \mbox{ if }\epsilon\in\llbracket L_{1}\rrbracket
\end{cases}\\
D_{c}\left(L\rightarrow f\right) & =D_{c}\left(L\right)\rightarrow f
\end{alignedat}
\]

\caption{\label{fig:Derivatives-of-parsing-expression-forms}Derivatives of
parsing expression forms}
\end{figure}

\begin{figure}[tb]
\[
\begin{alignedat}{1}\delta\left(\emptyset\right) & =\mathbf{false}\\
\delta\left(\epsilon_{s}\right) & =\mathbf{true}\\
\delta\left(c\right) & =\mathbf{false}\\
\delta\left(L_{1}\cup L_{2}\right) & =\delta\left(L_{1}\right)\mbox{ or }\delta\left(L_{2}\right)\\
\delta\left(L_{1}\circ L_{2}\right) & =\delta\left(L_{1}\right)\mbox{ and }\delta\left(L_{2}\right)\\
\delta\left(L\rightarrow f\right) & =\delta\left(L\right)
\end{alignedat}
\]

\caption{\label{fig:Nullability}Nullability of parsing expression forms}
\end{figure}

Because the derivative of a concatenation $L_{1}\circ L_{2}$ depends
on whether $\llbracket L_{1}\rrbracket$ contains the empty string,
$\epsilon$, we define a nullability function, $\delta\left(L\right)$,
in \prettyref{fig:Nullability} such that it returns boolean true
or false when $\llbracket L\rrbracket$ respectively contains $\epsilon$
or does not. The null language, $\emptyset$, contains nothing, and
the single-token language, $c$, contains only the word consisting
of the token $c$. Because neither of these languages contain $\epsilon$,
their nullability is false. Conversely, the $\epsilon_{s}$ language
contains only the $\epsilon$ word, so its nullability is true. The
union of two languages contains $\epsilon$ if either of its children
contains $\epsilon$, so the union is nullable if either $L_{1}$
or $L_{2}$ is nullable. Given how the semantics of the concatenation
$L_{1}\circ L_{2}$ are defined in \prettyref{fig:parsing-expression-forms},
in order for $L_{1}\circ L_{2}$ to contain $\epsilon$, there must
exist a $uv$ equal to $\epsilon$. This happens only when $u$ and
$v$ are both $\epsilon$, so a concatenation is nullable if and only
if both its children are nullable. Finally, the words in a reduction
$L\rightarrow f$ are those words in $L$, so its nullability is the
nullability of $L$.

\subsection{Derivatives of Context-free Languages\label{sub:Derivatives-of-context-free}}

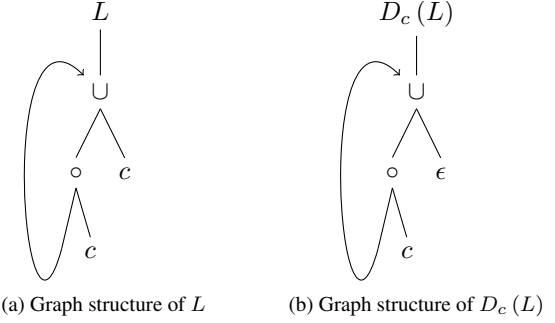
\begin{figure}[tb]
\hfill{}\subfloat[\label{fig:cyclic-graph-of-L}Graph structure of $L$]{\begin{tikzpicture}
\tikzset{level distance=3em}
\clip (-2,0.25) rectangle (2,-3.5);
\Tree [.$L$ [.\node(root){$\cup$}; [.$\circ$ \node(cycle){}; $c$ ] $c$ ] ]
\path[draw,->] (cycle.north) .. controls +(255:2) and +(north west:2) .. (root);
\end{tikzpicture}

}\hfill{}\subfloat[\label{fig:cyclic-graph-of-DcL}Graph structure of $D_{c}\left(L\right)$]{\begin{tikzpicture}
\tikzset{level distance=3em}
\clip (-2,0.25) rectangle (2,-3.5);
\Tree [.$D_c\left(L\right)$ [.\node(root){$\cup$}; [.$\circ$ \node(cycle){}; $c$ ] $\epsilon$ ] ]
\path[draw,->] (cycle.north) .. controls +(255:2) and +(north west:2) .. (root);
\end{tikzpicture}}\hfill{}

\caption{An example grammar and its derivative}
\end{figure}

\subsubsection{Representation}

\citet{Might:2011:PDF:2034773.2034801} generalize from taking derivatives
of regular expressions to taking derivatives of full CFGs. In so doing,
\citet{Might:2011:PDF:2034773.2034801} do not use typical CFGs but
do use an equivalent construction.

First, instead of non-terminals mapping to zero or more sequences
of terminals and non-terminals, they map to a parsing expression.
This is akin to the parsing expressions in \citet{Ford:2004:10.1145/964001.964011}.
For example, any CFG can be converted to this expression form by converting
productions of the form 
\[
N\Coloneqq X_{11}\cdots X_{1m_{1}}\mid\cdots\mid X_{n1}\cdots X_{nm_{n}}
\]
to 
\[
N=X_{11}\circ\ldots\circ X_{1m_{1}}\cup\cdots\cup X_{n1}\circ\ldots\circ X_{nm_{n}}
\]

Second, in the data structures representing grammars, instead of using
explicitly named non-terminals, parsing expressions point directly
to the non-terminal's parsing expression. For example, we may have
a grammar like the following, where $c$ is some token.
\[
L=\left(L\circ c\right)\cup c
\]

\citet{Might:2011:PDF:2034773.2034801} represent this as the data
structure in \prettyref{fig:cyclic-graph-of-L} with the edge where
$L$ refers back to itself, forming a cycle in the data structure.
For the purposes of discussion, though, we will refer to non-terminals
and their names even though the actual representation uses direct
pointers instead of non-terminal names.

\subsubsection{Computation}

A complication of this representation occurs when taking a derivative.
If we blindly follow the rules in \prettyref{fig:Derivatives-of-parsing-expression-forms},
then the derivative of $L$ by $c$ is the following.
\[
D_{c}(L)=\left(D_{c}(L)\circ c\right)\cup\epsilon
\]
This $D_{c}\left(L\right)$ is recursive, so to compute $D_{c}\left(L\right)$,
we must already know $D_{c}\left(L\right)$!

\citet{Might:2011:PDF:2034773.2034801} solve this problem with two
measures. First, they memoize their derivation function, \code{derive},
by keeping a table containing, for each set of arguments, the results
that it returns. When \code{derive} is called, if the table already
contains an entry for its arguments, $\code{derive}$ uses the result
in the entry instead of re-computing the derivative. Otherwise, \code{derive}
performs the calculation as usual and, before returning, stores its
result in the memoization table so it can be used by any further calls
with those same arguments. If the same derivative is needed multiple
times, this ensures it is computed only once.

On its own, memoization does not prevent infinite loops due to cycles,
however, because $\code{derive}$ adds memoization table entries only
after it finishes computing. This is where a second measure comes
into play. Before doing any recursive calls, $\code{derive}$ puts
a partially constructed grammar node that is missing its children
into the memoization table. For the example of $D_{c}\left(L\right)$,
we know without having to recur into $L$'s children that the resultant
node is a $\cup$ . Thus, we can place such a node in the memoization
table before computing its children and temporarily mark its children
as unknown. Any recursive calls to $D_{c}\left(L\right)$ can find
and use this memoized result even though the derivatives of its children
have not yet been calculated. When the derivatives for the node's
children return, we update the children of the output node to point
to the results of those derivatives. This process can be viewed as
a sort of lazy computation and results in a graph structure like in
\prettyref{fig:cyclic-graph-of-DcL}.

Like with the derivative, computing nullability must also deal with
cycles in grammars. However, memoization alone is not sufficient here.
A cycle means the derivative of some node must point to one of its
ancestors. With nullability, though, we must not only compute the
nullability of an ancestor but also inspect its value so we can compute
the nullability of the current node. This turns nullability into a
least fixed point problem over the lattice of booleans. \citet{Might:2011:PDF:2034773.2034801}
implement this with a naive algorithm that initially assumes all nodes
are not nullable and then recomputes the nullability of all nodes
reachable from a particular root node, using the current values for
each node. If, in the process, any nodes are newly discovered to be
nullable, then all reachable nodes are re-traversed and this process
is repeated until there are no more changes.

\subsection{Performance}

Despite PWD's simplicity and elegance, \citet{Might:2011:PDF:2034773.2034801}
report significant problems with its performance. Firstly, they compute
a worst-case bound of $O(2^{2n}G^{2})$ for a grammar of size $G$
and an input of size $n$. Despite this, they note that average parse
time seems to be linear in the length of the input. Unfortunately,
even with this apparent linear behavior, their parser is exceedingly
slow. For example, they report that a 31-line Python file took three
minutes to parse! Using an optimization they call compaction that
prunes branches of the derived grammars as they emerge, they report
that execution time for the 31-line input comes down to two seconds.
Still, this is exceedingly slow for such a small input.

\section{Complexity Analysis\label{sec:Complexity-Analysis}}

\citet{Might:2011:PDF:2034773.2034801} report an exponential bound
for their algorithm, but they never show it is a tight bound. On the
contrary, it turns out that PWD can, in fact, be implemented in cubic
time.

As mentioned before, at its core, PWD involves four recursive functions:
\code{nullable?}, \code{derive}, \code{parse-null}, and \code{parse}.
The \code{nullable?} and \code{derive} functions implement $\delta\left(L\right)$
and $D_{c}\left(L\right)$, respectively; the \code{parse-null} function
extracts the final AST; and \code{parse} implements the outer loop
over input tokens. In \prettyref{sub:Complexity-in-terms-of-grammar-nodes},
we observe that the running times of these functions are bounded by
the number of grammar nodes in the initial grammar plus the number
of grammar nodes constructed during parsing. Next, in \prettyref{sub:nodes-in-terms-of-length}
we discover that the total number of nodes constructed during parsing
is $O\mathopen{}\left(Gn^{3}\right)$, where $G$ is the size of the
initial grammar and $n$ is the length of the input. How to structure
this part of the proof is the essential insight in our analysis and
is based on counting unique names that we assign to nodes. When combined
with the results from \prettyref{sub:Complexity-in-terms-of-grammar-nodes},
this then leads to a cubic bound on the total runtime.

Throughout this section, let $G$ be the number of grammar nodes in
the initial grammar, let $g$ be the number of nodes created during
parsing, and let $n$ be the length of the input. Also, when analyzing
a memoized function, we consider the cost of the check to see if a
memoized result exists for a particular input to be part of the running
time of the caller instead of the callee.

\subsection{Total Running Time in Terms of Grammar Nodes\label{sub:Complexity-in-terms-of-grammar-nodes}}

First, consider \code{nullable?}, which computes a boolean value
for each parse node in terms of a least fixed point. The implementation
by \citet{Might:2011:PDF:2034773.2034801} iteratively re-traverses
the grammar until no new nodes can be proven \code{nullable?}. Such
an algorithm is quadratic in the number of nodes over which \code{nullable?}
is being computed because each traversal might update only one node.
However, a more intelligent algorithm that tracks dependencies between
nodes and operates over the boolean lattice can implement this function
in linear time, as shown in the following lemma.
\begin{lem}
\label{lem:nullable-and-empty-running-time}The sum of the running
times of all invocations of \code{nullable?} is $O\mathopen{}\left(G+g\right)\mathclose{}$.\end{lem}
\begin{proof}
The fixed point to calculate \code{nullable?} can be implemented
by a data-flow style algorithm \citep{Kildall:1973:UAG:512927.512945}
that tracks which nodes need their nullability reconsidered when a
given node is discovered to be nullable. Such an algorithm is linear
in the product of the height of the lattice for the value stored at
each node and the number of direct dependencies between nodes. In
this case, the lattice is over booleans and is of constant height.
Since each node directly depends on at most two children, the number
of dependencies is bounded by twice the number of nodes ever created.
\end{proof}
Next, we have \code{derive}. Since \code{derive} is memoized, it
is tempting to analyze it in terms of the nodes passed to it. However,
each node may have its derivative taken with multiple different input
tokens. The work done by \code{derive} thus depends on the number
of tokens by which each node is derived, so we can instead simplify
things by analyzing \code{derive} in terms of the nodes that it constructs.
\begin{lem}
\label{lem:derive-running-time}The sum of the running times of all
invocations of \code{derive} is $O\mathopen{}\left(G+g\right)\mathclose{}$.\end{lem}
\begin{proof}
Every call to \code{derive} that is not cached by memoization creates
at least one new node and, excluding the cost of recursive calls,
does $O\mathopen{}\left(1\right)\mathclose{}$ work. As a result,
the number of nodes created, $g$, is at least as great as the amount
of work done. Thus, the work done by all calls to \code{derive} is
$O\mathopen{}\left(g\right)\mathclose{}$ plus the work done by \code{nullable?}.
By \prettyref{lem:nullable-and-empty-running-time}, this totals to
$O\mathopen{}\left(G+g\right)\mathclose{}$.
\end{proof}
Next, we have \code{parse-null}. For this part of the proof, we assume
that ASTs use ambiguity nodes and a potentially cyclic graph representation.
This is a common and widely used assumption when analyzing parsing
algorithms. For example, algorithms like GLR \citep{Lang:1974:10.1007/3-540-06841-4_65}
and \citeauthor{early-phd} \citep{early-phd,Earley:1970:10.1145/362007.362035}
are considered cubic, but only when making such assumptions. Without
ambiguity nodes, the grammar \code{S -> S S | a | b} has an exponential
number of unique parses for strings of length $n$ that have no repeated
substrings of length greater than $\log_{2}n$. Many of those parses
share common sub-trees, so it does not take exponential space when
represented with ambiguity nodes. Our implementation is capable of
operating either with or without such a representation, but the complexity
result holds only with the assumption.

Under these assumptions, \code{parse-null} is a simple memoized function
over grammar nodes and thus is linear.
\begin{lem}
\label{lem:parse-null-running-time}The sum of the running times of
all invocations of \code{parse-null} is $O\mathopen{}\left(G+g\right)$.\end{lem}
\begin{proof}
Every call to \code{parse-null} that is not cached by memoization
does $O\left(1\right)$ work, excluding the cost of recursive calls.
There are at most $G+g$ such non-cached calls.
\end{proof}
\noindent Finally, we have the total running time of \code{parse}.
\begin{thm}
\label{thm:parse-running-time}The total running time of \code{parse}
is $O\mathopen{}\left(G+g\right)$.\end{thm}
\begin{proof}
The \code{parse} function calls \code{derive} for each input token
and, at the end, calls \code{parse-null} once. By \prettyref{lem:derive-running-time}
and \prettyref{lem:parse-null-running-time}, these together total
$O\mathopen{}\left(G+g\right)$.
\end{proof}

\subsection{Grammar Nodes in Terms of Input Length\label{sub:nodes-in-terms-of-length}}

\noindent All of the results in \prettyref{sub:Complexity-in-terms-of-grammar-nodes}
depend on $g$, the number of grammar nodes created during parsing.
If we look at the definition of $D_{c}\left(L\right)$ (i.e., \code{derive})
in \prettyref{fig:Derivatives-of-parsing-expression-forms}, most
of the clauses construct only a single node and use the children of
the input node only once each. When combined with memoization, for
a given input token, these clauses create at most the same number
of nodes as there are in the grammar for the result of the derivative
just before parsing that input token. On their own, these clauses
thus lead to the construction of only $Gn$ nodes.

However, the clause for a sequence node $L_{1}\circ L_{2}$, when
$L_{1}$ is nullable, uses $L_{2}$ twice. This duplication is what
led many to believe PWD was exponential; and indeed, without memoization,
it would be. In order to examine this more closely, we assign unique
names to each node. We choose these names such that each name is unique
to the derivative of a particular node with respect to a particular
token. Thus, the naming scheme matches the memoization strategy, and
the memoization of \code{derive} ensures that two nodes with the
same name are always actually the same node.
\begin{defn}
We give each node a unique name that is a string of symbols determined
by the following rules.\medskip{}

\begin{minipage}[t]{0.95\columnwidth}%
\begin{description}
\item [{Rule~5a:}] Nodes in the initial grammar are given a name consisting
of a single unique symbol distinct from that of any other node in
the grammar.
\item [{Rule~5b:}] When the node passed to \code{derive} has the name
$w$ and is a $\circ$ node containing a nullable left child, the
$\cup$ node created by \code{derive} is given the name $w\mathord{\bullet}c$
where $\mathord{\bullet}$ is a distinguished symbol that we use for
this purpose and $c$ is the token passed to \code{derive}.
\item [{Rule~5c:}] Any other node created by \code{derive} is given a
name of the form $wc$, where $w$ and $c$ are respectively the name
of the node passed to \code{derive} and the token passed to \code{derive}.\end{description}
\end{minipage}
\end{defn}
A $\circ$ node with a nullable left child has the special case of
Rule~5b because it is the only case where \code{derive} produces
more than one node, and we need to give these nodes distinct names.
These resultant nodes are a $\cup$ node and a $\circ$ node that
is the left child of the $\cup$ node. The introduction of the $\bullet$
symbol in the name of the $\cup$ node keeps this name distinct from
the name of the $\circ$ node.

As an example of these rules, \prettyref{fig:node-names-example}
shows the nodes and corresponding names for the nodes created when
parsing the following grammar. 
\[
L=\left(L\circ L\right)\cup c
\]

In this example, $c$ accepts any token; the initial names are $L$,
$M$, and $N$; and the input is $c_{1}c_{2}c_{3}c_{4}$. Each node
in \prettyref{fig:node-names-example} is labeled with its name in
a subscript, and children that point to already existing nodes are
represented with a box containing the name of that node. For example,
the root of the first grammar contains the node named $L$ as its
root, and the node named $M$ in that tree has $L$ as both its children.
The dotted arrows in this diagram show where concatenation causes
duplication. The node $M$ produces $Mc_{1}$, $Mc_{1}$ produces
$Mc_{1}\mathord{\bullet}c_{2}$ and $Mc_{1}c_{2}$, and so on.

A nice property of these rules can be seen if we consider node names
with their initial unique symbol and any $\mathord{\bullet}$ symbols
removed. The remaining symbols are all tokens from the input. Furthermore,
these symbols are added by successive calls to \code{derive} and
thus are substrings of the input. This lets us prove the following
lemma.
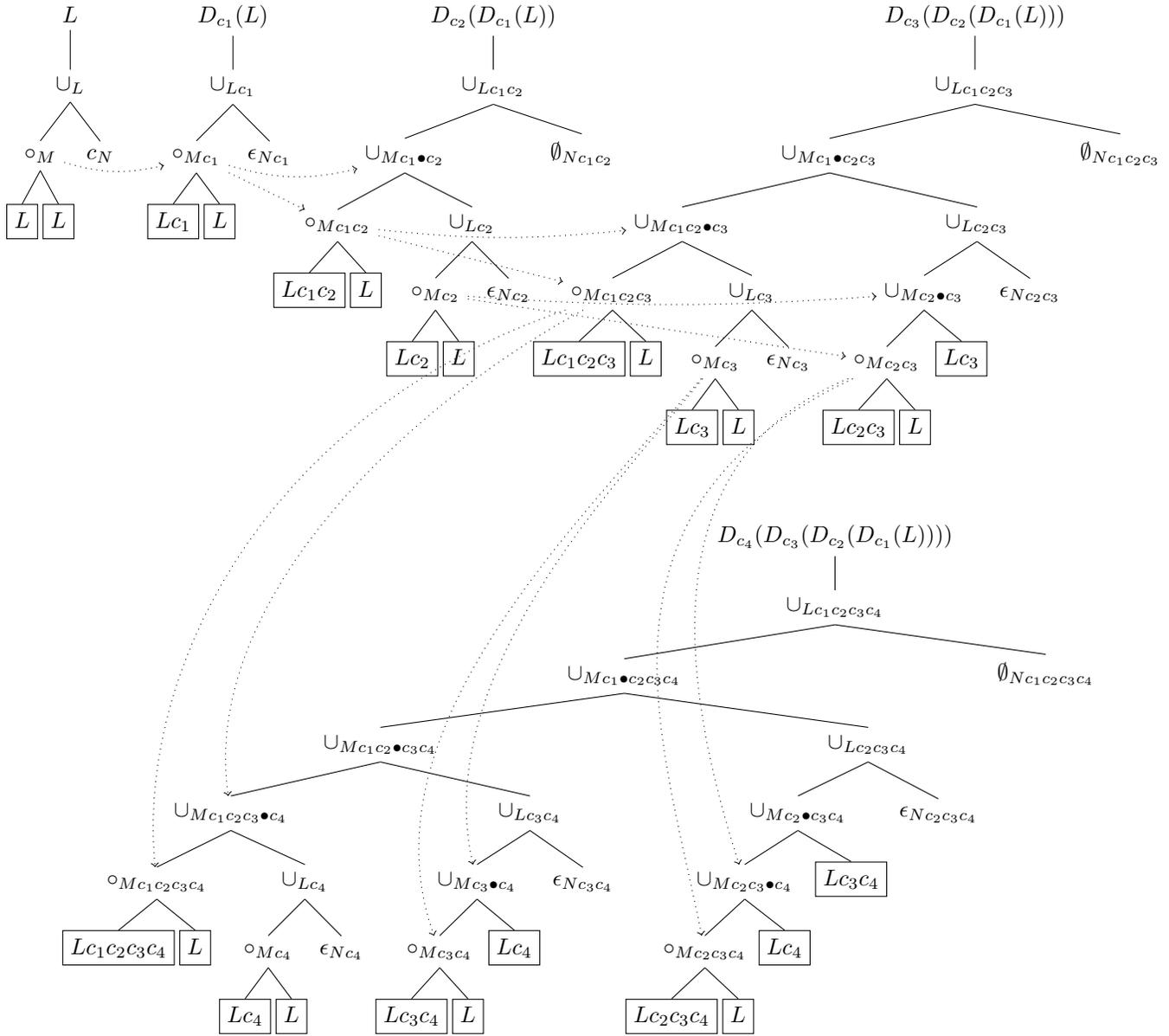
\begin{figure*}[tb]
\noindent \begin{center}
\begin{tikzpicture}
\tikzset{level distance=3em}
%\clip (-5,0.25) rectangle (5,-5);

\begin{scope}
\Tree [.$L$ [.$\cup_{L}$
 [.\node(M){$\circ_{M}$}; \node[draw]{$L$\vphantom{$c_0$}}; \node[draw]{$L$\vphantom{$c_0$}}; ] $c_{N}$ ] ]
\end{scope}

\begin{scope}[xshift=2.5cm]
\Tree [.$D_{c_1}(L)$ 
 [.$\cup_{Lc_1}$
  [.\node(Mc1){$\circ_{Mc_1}$};
   \node[draw]{$Lc_1$};
   \node[draw]{$L$\vphantom{$c_0$}};
  ]
  $\epsilon_{Nc_1}$
 ]
]
\end{scope}

\begin{scope}[xshift=6.5cm]
\Tree [.$D_{c_2}(D_{c_1}(L))$
 [.$\cup_{Lc_1c_2}$
  [.\node(Mc1Bc2){$\cup_{Mc_1\bullet{}c_2}$};
   [.\node(Mc1c2){$\circ_{Mc_1c_2}$};
    \node[draw]{$Lc_1c_2$};
    \node[draw]{$L$\vphantom{$c_0$}};
   ]
   [.$\cup_{Lc_2}$
    [.\node(Mc2){$\circ_{Mc_2}$};
     \node[draw]{$Lc_2$};
     \node[draw]{$L$\vphantom{$c_0$}};
    ]
    $\epsilon_{Nc_2}$
   ]
  ]
 $\emptyset_{Nc_1c_2}$
 ]
]
\end{scope}

\begin{scope}[xshift=13.9cm]
\Tree [.$D_{c_3}(D_{c_2}(D_{c_1}(L)))$
 [.$\cup_{Lc_1c_2c_3}$
  [.$\cup_{Mc_1\bullet{}c_2c_3}$
   [.\node(Mc1c2Bc3){$\cup_{Mc_1c_2\bullet{}c_3}$};
    [.\node(Mc1c2c3){$\circ_{Mc_1c_2c_3}$};
     \node[draw]{$Lc_1c_2c_3$};
     \node[draw]{$L$\vphantom{$c_0$}};
    ]
    [.$\cup_{Lc_3}$
     [.\node(Mc3){$\circ_{Mc_3}$};
      \node[draw]{$Lc_3$};
      \node[draw]{$L$\vphantom{$c_0$}};
     ]
     $\epsilon_{Nc_3}$
    ]
   ]
   [.$\cup_{Lc_2c_3}$
    [.\node(Mc2Bc3){$\cup_{Mc_2\bullet{}c_3}$};
     [.\node(Mc2c3){$\circ_{Mc_2c_3}$};
      \node[draw]{$Lc_2c_3$};
      \node[draw]{$L$\vphantom{$c_0$}};
     ]
     \node[draw]{$Lc_3$};
    ]
    $\epsilon_{Nc_2c_3}$
   ]
  ]
 $\emptyset_{Nc_1c_2c_3}$
 ]
]
\end{scope}

\begin{scope}[xshift=11.75cm,yshift=-8cm]
\Tree [.$D_{c_4}(D_{c_3}(D_{c_2}(D_{c_1}(L))))$
 [.$\cup_{Lc_1c_2c_3c_4}$
  [.$\cup_{Mc_1\bullet{}c_2c_3c_4}$
   [.$\cup_{Mc_1c_2\bullet{}c_3c_4}$
    [.\node(Mc1c2c3Bc4){$\cup_{Mc_1c_2c_3\bullet{}c_4}$};
     [.\node(Mc1c2c3c4){$\circ_{Mc_1c_2c_3c_4}$};
      \node[draw]{$Lc_1c_2c_3c_4$};
      \node[draw]{$L$\vphantom{$c_0$}};
     ]
     [.$\cup_{Lc_4}$
      [.$\circ_{Mc_4}$
       \node[draw]{$Lc_4$};
       \node[draw]{$L$\vphantom{$c_0$}};
      ]
      $\epsilon_{Nc_4}$
     ]
    ]
    [.$\cup_{Lc_3c_4}$
     [.\node(Mc3Bc4){$\cup_{Mc_3\bullet{}c_4}$};
      [.\node(Mc3c4){$\circ_{Mc_3c_4}$};
       \node[draw]{$Lc_3c_4$};
       \node[draw]{$L$\vphantom{$c_0$}};
      ]
      \node[draw]{$Lc_4$};
     ]
     $\epsilon_{Nc_3c_4}$
    ]
   ]
   [.$\cup_{Lc_2c_3c_4}$
    [.$\cup_{Mc_2\bullet{}c_3c_4}$
     [.\node(Mc2c3Bc4){$\cup_{Mc_2c_3\bullet{}c_4}$};
      [.\node(Mc2c3c4){$\circ_{Mc_2c_3c_4}$};
       \node[draw]{$Lc_2c_3c_4$};
       \node[draw]{$L$\vphantom{$c_0$}};
      ]
      \node[draw]{$Lc_4$};
     ]
     \node[draw]{$Lc_3c_4$};
    ]
    $\epsilon_{Nc_2c_3c_4}$
   ]
  ]
 $\emptyset_{Nc_1c_2c_3c_4}$
 ]
]
\end{scope}

\path[draw,dotted,->] (M) to[bend right=15] (Mc1);
\path[draw,dotted,->] (Mc1) to[bend right=15] (Mc1Bc2);
\path[draw,dotted,->] (Mc1) to[] (Mc1c2);

\path[draw,dotted,->] (Mc1c2) to[bend right=6] (Mc1c2Bc3);
\path[draw,dotted,->] (Mc1c2) to[] (Mc1c2c3);
\path[draw,dotted,->] (Mc2) to[bend right=4] (Mc2Bc3);
\path[draw,dotted,->] (Mc2) to[] (Mc2c3);

\path[draw,dotted,->] (Mc1c2c3) .. controls (5,-6) and (2,-9) .. (Mc1c2c3Bc4);
\path[draw,dotted,->] (Mc1c2c3) .. controls (3,-6) and (1,-10) .. (Mc1c2c3c4);
\path[draw,dotted,->] (Mc3) .. controls (8.5,-7) and (5.5,-10) .. (Mc3Bc4);
\path[draw,dotted,->] (Mc3) .. controls (7.5,-8) and (4.5,-11) .. (Mc3c4);
\path[draw,dotted,->] (Mc2c3) .. controls (9.5,-6.5) and (9,-9) .. (Mc2c3Bc4);
\path[draw,dotted,->] (Mc2c3) .. controls (7.75,-7.5) and (9.1,-11) .. (Mc2c3c4);

\end{tikzpicture}
\par\end{center}

\caption{\label{fig:node-names-example}Worst-case behavior of PWD. Nodes are
annotated with their names in subscripts.}
\end{figure*}

\begin{lem}
The number of strings of symbols consisting of node names with their
initial unique symbols and any $\mathord{\bullet}$ symbols removed
is $O\mathopen{}\left(n^{2}\right)\mathclose{}$ .\end{lem}
\begin{proof}
These strings are all substrings of the input. \citet{counting-substrings}
count the number of such substrings and show that, unsurprisingly,
it is $O\mathopen{}\left(n^{2}\right)\mathclose{}$. At an intuitive
level, this is because the number of positions where these substrings
can start and end in the input are both linear in $n$.
\end{proof}
In \prettyref{fig:node-names-example}, this can be seen by the fact
that the $c_{1}$, $c_{2}$, $c_{3}$, and $c_{4}$ occurring in node
names are always in increasing, consecutive ranges such as $c_{1}c_{2}c_{3}$
in $Mc_{1}c_{2}\mathord{\bullet}c_{3}$ or $c_{2}c_{3}$ in $Nc_{2}c_{3}$.

Another nice property of names is that they all contain at most one
occurrence of $\bullet$. This turns out to be critical. At an intuitive
level, this implies that the $\cup$ node involved in a duplication
caused by a $\circ$ node is never involved in another duplication.
\begin{lem}
Each node name contains at most one occurrence of the $\mathord{\bullet}$
symbol.\end{lem}
\begin{proof}
According to Rule~5b, a $\mathord{\bullet}$ symbol is put in the
name of only those $\cup$ nodes that come from taking the derivative
of a $\circ$ node. Further derivatives of these $\cup$ nodes can
produce only more $\cup$ nodes, so Rule~5b, which applies only to
$\circ$ nodes, cannot apply to any further derivatives of those $\cup$
nodes. Thus, once a $\mathord{\bullet}$ symbol is added to a name,
another one cannot be added to the name. 
\end{proof}
This property can be seen in \prettyref{fig:node-names-example} where
no name contains more than one $\mathord{\bullet}$, and every node
that does contain $\mathord{\bullet}$ is a $\cup$ node.

This then implies that every name is either of the form $Nw$ or $Nu\mathord{\bullet}v$,
where $N$ is the name of an initial grammar node and both $w$ and
$uv$ are substrings of the input. As a result, we can bound the number
of possible names with the following theorem.
\begin{thm}
\label{thm:number-of-nodes-created}The total number of nodes constructed
during parsing is $O\mathopen{}\left(Gn^{3}\right)\mathclose{}$.\end{thm}
\begin{proof}
In a name of the form $Nw$ or $Nu\mathord{\bullet}v$, the number
of possible symbols for $N$ is the size of the initial grammar, $G$.
Also, the number of possible words for $w$ or $uv$ is bounded by
the number of unique subwords in the input, which is $O\mathopen{}\left(n^{2}\right)$.
Finally, the number of positions at which $\mathord{\bullet}$ may
occur within those subwords is $O\mathopen{}\left(n\right)\mathclose{}$.
The number of unique names, and consequently the number of nodes created
during parsing, is the product of these: $O\mathopen{}\left(Gn^{3}\right)\mathclose{}$.
\end{proof}

\subsection{Running Time in Terms of Input Length}

\noindent Finally, we can conclude that the running time of parsing
is cubic in the length of the input.
\begin{thm}
\label{thm:running-time-of-parsing}The running time of \code{parse}
is $O\mathopen{}\left(Gn^{3}\right)\mathclose{}$.\end{thm}
\begin{proof}
Use $O\mathopen{}\left(Gn^{3}\right)\mathopen{}$ in \prettyref{thm:number-of-nodes-created}
for $g$ in \prettyref{thm:parse-running-time}.
\end{proof}
Note that this analysis does not assume the use of the process that
\citet{Might:2011:PDF:2034773.2034801} call compaction. Nevertheless,
it does hold, in that case, if compaction rules are applied only when
a node is constructed and only locally at the node being constructed.
The extra cost of compaction is thus bounded by the number of nodes
constructed, and compaction only ever reduces the number of nodes
constructed by other parts of the parser.

\section{Improving Performance in Practice\label{sec:Practical-Performance}}

Given that PWD has a cubic running time instead of the exponential
conjectured in \citet{Might:2011:PDF:2034773.2034801}, the question
remains of why their implementation performed so poorly and whether
it can be implemented more efficiently. To investigate this, we reimplemented
PWD from scratch and built up the implementation one part at a time.
We measured the running time as each part was added and adjusted our
implementation whenever a newly added part significantly slowed down
the implementation. \prettyref{sub:Benchmarks} reports the final
performance of the resulting parser. Aside from low-level needs to
choose efficient data structures, we found three major algorithmic
improvements, which are discussed in \prettyref{sub:Computing-fixed-points},
\prettyref{sub:Compaction}, and \prettyref{sub:Hash-tables}.

The resulting parser implementation remains rather simple and easily
read. The optimization of the fixed-point computation for \code{nullable?}
(\prettyref{sub:Computing-fixed-points}) takes 24~lines of Racket
code, including all helpers. Compaction (\prettyref{sub:Compaction})
is implemented using smart constructors for each form that test if
they are constructing a form that can be reduced. This takes 50 lines
of code due to each constructor needing to have a clause for each
child constructor with which it could reduce. Finally, single-entry
memoization (\prettyref{sub:Hash-tables}) requires changing only
the helpers that implement memoization, which does not increase the
complexity or size of the resulting code. With all of these optimizations
implemented, the core code is 62~lines of Racket code with an additional
76~lines of code for helpers.

The complete implementation can be downloaded from:
\begin{lyxcode}
http://www.bitbucket.com/ucombinator/derp-3
\end{lyxcode}

\subsection{Benchmarks\label{sub:Benchmarks}}

In order to test the performance of our implementation of PWD, we
ran our parser on the files in the Python Standard Library version
3.4.3 \citep{Python343} using a grammar derived from the Python 3.4.3
specification \citep{Python343grammar}. The Python Standard Library
includes 663 Python files, which have sizes of up to 26,125 tokens.

We compared our parser against three parsers. The first one used the
original PWD implementation \citep{derp1-implementation}. The second
one used the \code{parser-tools/cfg-parser} library \citep{cfg-parser}
that comes with Racket 6.1.1 \citep{racket}. The third one used Bison
version 3.0.2 \citep{bison}.

In order to have a fair comparison against the original PWD implementation,
our parser was written in Racket. For compatibility with \code{parser-tools/cfg-parser}
and Bison, we modified our grammar to use traditional CFG productions
instead of the nested parsing expressions supported by PWD and used
by the Python grammar specification. The resulting grammar contained
722 productions.

The \code{parser-tools/cfg-parser} library uses a variant of the
\citeauthor{early-phd} parsing algorithm \citep{early-phd,Earley:1970:10.1145/362007.362035},
so it may not perform as well as other GLR parsers \citep{Lang:1974:10.1007/3-540-06841-4_65}.
Nevertheless, we used it because we were not able to locate a suitable
GLR parser for Racket.

In order to compare against a more practical GLR parser, we included
a Bison-based parser. We ran Bison in GLR mode, as the grammar resulted
in 92 shift/reduce and 4~reduce/reduce conflicts. However, as the
Bison-based parser is written in C and the improved PWD parser is
written in Racket, the Bison-based parser has an extra performance
boost that the improved PWD implementation does not have.

We ran the tests with Racket~6.1.1 and GCC~4.9.2 on a 64-bit, 2.10~GHz
Intel Core i3-2310M running Ubuntu 15.04. Programs were limited to
8000~MB of RAM via \code{ulimit}. We tokenized files in advance
and loaded those tokens into memory before benchmarking started, so
only parsing time was measured when benchmarking. For each file, we
computed the average of ten rounds of benchmarking that were run after
at least three warm-up rounds. However, the original PWD was so slow
that we could only do three rounds of benchmarking for that implementation.
Each round parsed the contents of the file multiple times, so the
run time lasted at least one second to avoid issues with clock quantization.
We cleared memoization tables before the start of each parse. A small
number of files exceeded 8000~MB of RAM when parsed by the original
PWD or \code{parser-tools/cfg-parser} and were terminated early.
We omit the results for those parsers with those files. This did not
happen with the improved PWD and Bison, and the results from those
parsers on those files are included. The final results are presented
in \prettyref{fig:Performance-of-parsers} and are normalized to measure
parse time per input token.

As reported in \citet{Might:2011:PDF:2034773.2034801}, PWD appears
to run in linear time, in practice, with a constant time per token.
However, our improved parser runs on average 951~times faster than
that by \citet{Might:2011:PDF:2034773.2034801}. It even runs 64.6~times
faster than the parser that uses the \code{parser-tools/cfg-parser}
library. As expected, our implementation ran slower than the Bison-based
parser, but by only a factor of 25.2. This is quite good, considering
how simple our implementation is and the differences in the implementations'
languages. We suspect that further speedups could be achieved with
a more efficient implementation language.

In the remainder of this section, we explain the main high-level algorithmic
techniques we discovered that achieve this performance.

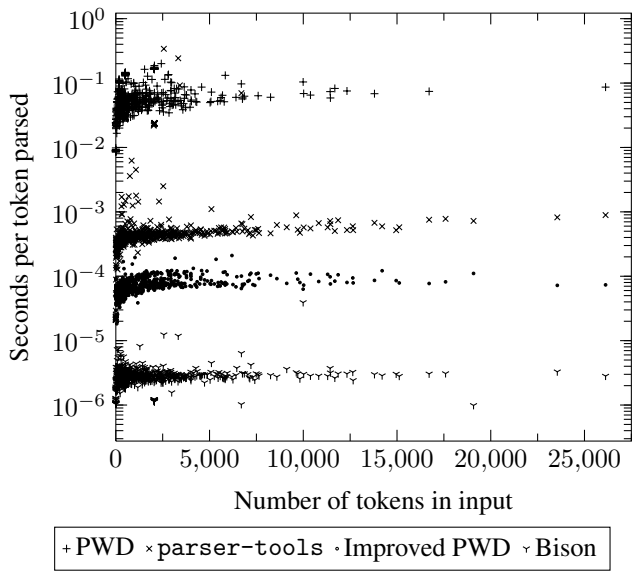
\begin{figure}[tb]
\noindent \begin{centering}
\noindent \begin{center}
\begin{tikzpicture} 
\begin{axis}[scaled ticks=false, enlarge x limits=false, xmax=27500, ymode=log,
 every tick/.style=black, minor x tick num=1, ytickten={-20,...,20},
 xlabel={Number of tokens in input},
 ylabel={Seconds per token parsed},
 legend entries={PWD\phantom{.},\code{parser-tools}\phantom{.}, Improved PWD\phantom{.}, Bison},
 legend cell align=left,
 legend columns=4,
 legend style={at={(0.42,-0.2)},anchor=north}]
\addplot[only marks, mark size=1.5pt, mark=+] table [
 x expr=\thisrowno{10},
 y expr=\thisrowno{2}/\thisrowno{10}] {data/derp1-times3.data};
\addplot[only marks, mark size=1.5pt, mark=x] table [
 x expr=\thisrowno{10},
 y expr=\thisrowno{2}/\thisrowno{10}] {data/derp2-times.data};
\addplot[only marks, mark size=0.5pt] table [
 x expr=\thisrowno{10},
 y expr=\thisrowno{2}/\thisrowno{10}] {data/derp3-times2.data};
\addplot[only marks, mark size=1.5pt, mark=Mercedes star flipped] table [
 x expr=\thisrowno{4},
 y expr=\thisrowno{2}/\thisrowno{4}] {data/yacc-times.data};
\end{axis}
\end{tikzpicture}
\par\end{center}
\par\end{centering}

\caption{\label{fig:Performance-of-parsers}Performance of various parsers}
\end{figure}

\subsection{Computing Fixed Points\label{sub:Computing-fixed-points}}

The \code{nullable?} function is defined in terms of a least fixed
point. The implementation in \citet{Might:2011:PDF:2034773.2034801}
computes this by repeatedly traversing over all grammar nodes. If
the computed nullability of any node changes during that traversal,
all of the nodes are traversed again. This continues until there are
no more changes.

This is a fairly naive method of computing a fixed point and is quadratic
in the number of nodes in the grammar as each re-traversal may update
only one node that then triggers another re-traversal. A more efficient
method uses ideas from data-flow analysis \citep{Kildall:1973:UAG:512927.512945}
and tracks which nodes depend on which others. When the computed nullability
of a node changes, only those nodes that depend on that node are revisited.

While the tracking of dependencies does incur an overhead, we can
minimize this by tracking dependencies only after discovering cycles
that prevent us from immediately computing the result. In other cases,
we directly compute nullability with a simple recursive traversal.

We can further improve the performance of \code{nullable?} by distinguishing
between nodes that are definitely not nullable and those that are
merely assumed to be not nullable because the fixed point has not
yet shown them to be nullable.

Assumed-not-nullable and definitely-not-nullable nodes behave almost
exactly alike except that we may re-traverse assumed-not-nullable
nodes but never re-traverse definitely-not-nullable nodes. This is
because definitely-not-nullable nodes have their final value, while
assumed-not-nullable nodes might not.

In many types of fixed-point problems, this is not an important distinction
because there is usually no way to distinguish between these types
of nodes. However, when computing nullability, we can take advantage
of this because the computation of nullability is not done only once.
Rather, it is called multiple times on different nodes by different
executions of \code{derive}. Within each of these fixed points, only
nodes reachable from the node passed to the initial call to \code{nullable?}
by \code{derive} have their nullability computed. Later calls to
\code{nullable?} may examine different nodes, but when they examine
nodes already examined in a previous call to \code{nullable?} from
\code{derive}, they can reuse information from that previous call.
Specifically, not only are nodes that are discovered by previous fixed
points to be nullable still nullable, but nodes that are assumed-not-nullable
at the end of a previous fixed-point calculation are now definitely-not-nullable.
This is because the nodes that could cause them to be nullable are
already at a value that is a fixed point and will not change due to
further fixed-point calculations.

We take advantage of this by marking nodes visited by \code{nullable?}
with a label that is unique to the call in \code{derive} that started
the nullability computation. Then, any nodes still assumed-not-nullable
that are marked with a label from a previous call are treated as definitely-not-nullable.

The end result of these optimizations is a significant reduction in
the number of calls to \code{nullable?}. In \prettyref{fig:calls-to-nullable},
we plot the number of calls to \code{nullable?} in our implementation
relative to that of \citet{Might:2011:PDF:2034773.2034801}. On average,
the new implementation has only 1.5\% of the calls to \code{nullable}
as that of \citet{Might:2011:PDF:2034773.2034801}.

\subsection{Compaction\label{sub:Compaction}}

\citet{Might:2011:PDF:2034773.2034801} report that a process that
they call compaction improves the performance of parsing by a factor
of about 90. We found similar results in our implementation, and the
benchmarks in \prettyref{fig:Performance-of-parsers} use compaction.
However, we also discovered improvements to this process.

First, we keep the following reduction rules from \citet{Might:2011:PDF:2034773.2034801}
with no changes. The first three rules take advantage of the fact
that $\emptyset$ is the identity of $\cup$ and the annihilator of
$\circ$. The last three rules move the operations involved in producing
an AST out of the way to expose the underlying grammar nodes. 
\[
\begin{alignedat}{1}\emptyset\cup p & \Rightarrow p\\
p\cup\emptyset & \Rightarrow p\\
\emptyset\circ p & \Rightarrow\emptyset\\
\epsilon_{s}\circ p & \Rightarrow p\rightarrow\lambda u.\left(s,u\right)\\
\epsilon_{s}\rightarrow f & \Rightarrow\epsilon_{\left(f\,s\right)}\\
\left(p\rightarrow f\right)\rightarrow g & \Rightarrow p\rightarrow\left(g\circ f\right)
\end{alignedat}
\]

To these rules, we add the following reductions, which were overlooked
in \citet{Might:2011:PDF:2034773.2034801}.
\[
\begin{alignedat}{1}\emptyset\rightarrow f & \Rightarrow\emptyset\\
\epsilon_{s_{1}}\cup\epsilon_{s_{2}} & \Rightarrow\epsilon_{s_{1}\cup s_{2}}
\end{alignedat}
\]
We also omit the following reduction used by \citet{Might:2011:PDF:2034773.2034801},
as it is covered by the reductions for $\epsilon_{s}\circ p$ and
$\left(p\rightarrow f\right)\rightarrow g$.
\[
\begin{alignedat}{1}\left(\epsilon_{s}\circ p\right)\rightarrow f & \Rightarrow p\rightarrow\end{alignedat}
\lambda u.f\,\left(u,s\right)
\]

The reader may notice that these laws are very similar to the laws
for Kleene algebras \citep{Kozen:1994:10.1006/inco.1994.1037}. If
we ignore the generated parse trees and consider only string recognition,
parsing expressions are Kleene algebras. The identities for compaction
have one subtle difference from those for Kleene algebras, however.
They must preserve the structure of the resulting parse tree, and
several of the identities insert reductions ($\rightarrow$) to do
this.

\begin{figure}[tb]
\noindent \begin{centering}
\noindent \begin{center}
\begin{tikzpicture}
\begin{axis}[scaled ticks=false, enlarge x limits=false, xmax=27500, ymin=0, ymax=0.1, %ymode=log,
 %legend entries={Old,New},
 %legend pos=south east, legend cell align=left
 xlabel={Number of tokens in input},
 ylabel={Relative number of calls to \code{nullable?}},
 every tick/.style=black, minor x tick num=1, 
 ytick={0.0,0.02,0.04,0.06,0.08,0.10},
 minor y tick num=1,
 yticklabel={\pgfmathparse{\tick*100}\pgfmathprintnumber{\pgfmathresult}\%}
]
\addplot[only marks, mark size=0.5pt] table [
 x expr=\thisrowno{0},
 y expr=(\thisrowno{7}-\thisrowno{11}+\thisrowno{9}-(10592-4526+4563))/\thisrowno{5}] {data/nullable2.data};
% This offset is due to the call to nullable when preparing the grammar before parsing
\end{axis}
\end{tikzpicture}
\par\end{center}
\par\end{centering}

\caption{\label{fig:calls-to-nullable}Number of calls to \protect\code{nullable?}
in the improved PWD relative to the original PWD}
\end{figure}
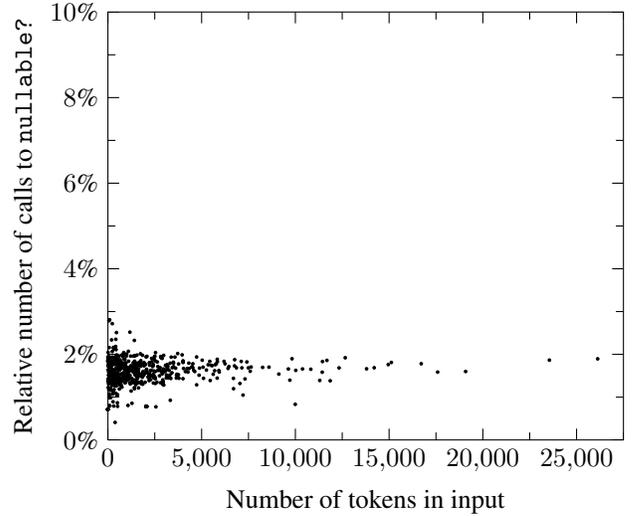

\subsubsection{Right-hand Children of Sequence Nodes\label{sub:Right-hand-children-of-seq}}

In our implementation, the following two reductions, which are used
by \citet{Might:2011:PDF:2034773.2034801}, are not used during parsing.
We omit these reductions because the forms on their left-hand sides
cannot occur during parsing unless the initial grammar contains them.
\begin{alignat*}{1}
p\circ\epsilon_{s} & \Rightarrow p\rightarrow\lambda u.\left(u,s\right)\\
p\circ\emptyset & \Rightarrow\emptyset
\end{alignat*}

\begin{thm}
While parsing, grammar nodes are never of the form $p\circ\epsilon_{s}$
or $p\circ\emptyset$ unless nodes in the initial grammar are of that
same form.\end{thm}
\begin{proof}
The derivative process changes only the left-hand child of a sequence
node. Thus, the right-hand child of a sequence node is always a copy
of the right-hand child of a sequence node from the initial grammar.
\end{proof}
We take advantage of this fact by using these reduction rules on the
initial grammar before parsing so that once parsing starts, we never
need to check for them again. This avoids the need to inspect the
right-hand children of sequence nodes during parsing and saves us
the cost of any resulting memory accesses or conditional branching.

\subsubsection{Canonicalizing Chains of Sequence Nodes}

Consider a grammar fragment, like in \prettyref{fig:seq-stack-left},
where $p_{1}$ is not nullable. When taking the derivative, only the
left-hand children of the sequence nodes are considered. Thus, none
of $p_{2},\cdots,p_{i-1},p_{i}$ are inspected by \code{derive},
though the sequence nodes containing them are traversed. We could
avoid the cost of this traversal if we restructured the grammar like
in \prettyref{fig:seq-stack-right} where the $f'$ function rearranges
the pairs in the resulting parse tree to match the AST produced by
\prettyref{fig:seq-stack-left}. As a result, \code{derive} would
traverse only two nodes, the reduction node and topmost sequence node,
instead of the $i$ nodes in \prettyref{fig:seq-stack-left}.

We can use compaction to try to optimize \prettyref{fig:seq-stack-left}
into \prettyref{fig:seq-stack-right} by adding the following reduction
rule, which implements associativity for sequence nodes.
\begin{multline*}
\left(p_{1}\circ p_{2}\right)\circ p_{3}\Rightarrow\left(p_{1}\circ\left(p_{2}\circ p_{3}\right)\right)\\
\rightarrow\lambda u.\left\{ \left(\left(t_{1},t_{2}\right),t_{3}\right)\mid\left(t_{1},\left(t_{2},t_{3}\right)\right)\in u\right\} 
\end{multline*}
However, this is not enough on its own. Depending on the order in
which nodes get optimized by this reduction rule, a reduction node
may be placed between neighboring sequence nodes that interferes with
further applications of this reduction rule. This can lead to structures
like in \prettyref{fig:seq-red-stack-left}. Indeed, our inspection
of intermediate grammars during parses revealed several examples of
this.

\begin{figure}[tb]
\hfill{}\subfloat[\label{fig:seq-stack-left}Left-associated sequence nodes]{\begin{tikzpicture}
\tikzset{level distance=3em}
\Tree [.$\circ$ [.$\circ$ \edge[draw=none] node {$\iddots$}; [.$\circ$ $p_1$ $p_2$ ] $p_{i-1}$ ] $p_i$ ]
\end{tikzpicture}

}\hfill{}\subfloat[\label{fig:seq-stack-right}Right-associated sequence nodes]{\begin{tikzpicture}
\tikzset{level distance=3em}
\Tree [.$\rightarrow f'$ [.$\circ$ $p_1$ [.$\circ$ $p_2$ \edge[draw=none] node {$\ddots$}; [.$\circ$ $p_{i-1}$ $p_i$ ] ] ] ]
\end{tikzpicture}

}\hfill{}

\caption{Examples of stacked sequence nodes}

\hfill{}\subfloat[\label{fig:seq-red-stack-left}Reductions mixed into sequence nodes
that prevent optimization]{\begin{tikzpicture}
\tikzset{level distance=3em}
\Tree [.$\rightarrow f_i$
 [.$\circ$
  [.$\rightarrow f_{i-1}$
   [.$\circ$
    \edge[draw=none] node {$\iddots$};
    [.$\rightarrow f_2$
     [.$\circ$
      [.$\rightarrow f_1$ $p_1$ ]
      $p_2$
     ]
    ]
    $p_{i-1}$
   ]
  ]
  $p_i$
 ]
]
\end{tikzpicture}

}\hfill{}\subfloat[\label{fig:seq-red-stack-right}Sequence nodes after moving the reductions]{\begin{tikzpicture}
\tikzset{level distance=3em}
\Tree [.$\rightarrow f_i'$
 [.$\rightarrow f_{i-1}'$ \edge[draw=none] node {$\vdots$};
 [.$\rightarrow f_2'$ [.$\rightarrow f_1'$
 [.$\circ$ [.$\circ$ \edge[draw=none] node {$\iddots$};
[.$\circ$ $p_1$ $p_2$ ] $p_{i-1}$ ] $p_i$ ] ] ] ] ]
\end{tikzpicture}

}\hfill{}

\caption{Examples of reductions mixed with sequence nodes}
\end{figure}

We resolve this by also adding the following rule that floats reduction
nodes above and out of the way of sequence nodes.
\begin{multline*}
\left(p_{1}\rightarrow f\right)\circ p_{2}\Rightarrow\\
\left(p_{1}\circ p_{2}\right)\rightarrow\lambda u.\left\{ \left(f\,\left\{ t_{1}\right\} ,t_{2}\right)\mid\left(t_{1},t_{2}\right)\in u\right\} 
\end{multline*}

If we apply this rule for $\left(p_{1}\rightarrow f\right)\circ p_{2}$
to the reduction nodes generated by applying the rule for $\left(p_{1}\circ p_{2}\right)\circ p_{3}$,
then we get \prettyref{fig:seq-red-stack-right} where each $f_{i}'$
does the work of $f_{i}$ at the appropriate point in the AST. If
we further use the rule for $\left(p\rightarrow f\right)\rightarrow g$
on the stack of reduction nodes in \prettyref{fig:seq-red-stack-right},
we get \prettyref{fig:seq-stack-right}, which allows derivatives
to be computed efficiently.

Note that there is also a version of this reduction rule for when
a reduction node is the right-hand instead of left-hand child of a
sequence. It is the following.
\begin{multline*}
p_{1}\circ\left(p_{2}\rightarrow f\right)\Rightarrow\\
\left(p_{1}\circ p_{2}\right)\rightarrow\lambda u.\left\{ \left(t_{1},f\,\left\{ t_{2}\right\} \right)\mid\left(t_{1},t_{2}\right)\in u\right\} 
\end{multline*}
 However, for the same reasons as in \prettyref{sub:Right-hand-children-of-seq},
we use this only on the initial grammar and not during parsing.

\subsubsection{Avoiding Separate Passes}

\citet{Might:2011:PDF:2034773.2034801} implement compaction as a
separate pass in between the calls to \code{derive} for successive
tokens. However, this means that nodes are traversed twice per token
instead of only once. To avoid this overhead, we immediately compact
nodes as they are constructed by \code{derive}. This results in two
complications.

The first complication is that we do not want to iterate these reductions
to reach a fixed point. We just do the reductions locally on the grammar
node being generated by \code{derive}. As a result, there may be
a few missed opportunities for applying reductions, but compactions
in later derivatives should handle these.

The second complication is that we must consider how to compact when
\code{derive} has followed a cycle in the grammar. The \code{derive}
function usually does not need to know anything about the derivatives
of the child nodes, which means that calculating these derivatives
can be deferred using the lazy techniques described in \prettyref{sub:Derivatives-of-context-free}.
This poses a problem with compaction though, as many of the rules
require knowing the structure of the child nodes. Like with the first
complication, we have \code{derive} punt on this issue. If inspecting
a child would result in a cycle, \code{derive} does not attempt to
compact. This design may miss opportunities to compact, but it allows
us to avoid the cost of a double traversal of the grammar nodes.

\subsection{Hash Tables and Memoization\label{sub:Hash-tables}}

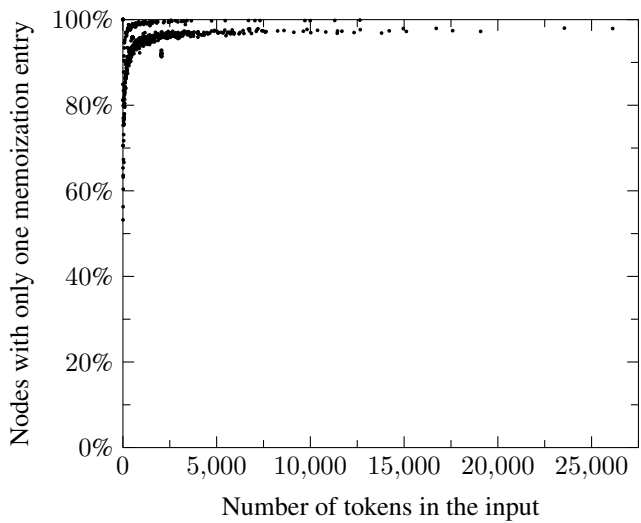
\begin{figure}[tb]
\noindent \begin{centering}
\noindent \begin{center}
\begin{tikzpicture}
\begin{axis}[
 enlarge x limits=false, xmax=27500,
 enlarge y limits=false, ymax=1, %{upper,value=0.05},
 xlabel={Number of tokens in the input},
 every tick/.style=black, minor x tick num=1,
 ylabel={Nodes with only one memoization entry}, minor y tick num=1,
 scaled ticks=false,
 ymin=0, yticklabel={\pgfmathparse{\tick*100}\pgfmathprintnumber{\pgfmathresult}\%},
 legend pos=south east, legend cell align=left]
%\addplot[only marks, mark size=0.5pt] table [
% x expr=\thisrowno{0},
% y expr=\thisrowno{2}/(\thisrowno{2}+\thisrowno{3}+\thisrowno{4})] {data/cache-sizes.data};
%\addplot[only marks, mark size=1.5pt, mark=+] table [
% x expr=\thisrowno{0},
% y expr=(\thisrowno{2}+\thisrowno{3})/(\thisrowno{2}+\thisrowno{3}+\thisrowno{4})] {data/cache-sizes.data};
%\addplot[only marks, mark size=0.5pt] table [
% x expr=\thisrowno{0},
% y expr=(\thisrowno{5}+\thisrowno{8})/(\thisrowno{5}+\thisrowno{8}+\thisrowno{2})] {data/cache-sizes2.data};
\addplot[only marks, mark size=0.5pt] table [
 x expr=\thisrowno{8},
 y expr=(\thisrowno{7})/(\thisrowno{7}+\thisrowno{1})] {data/cache-sizes3.data};
%\addplot[domain=0:30000,help lines] {0.99};
%\addplot[domain=0:30000,help lines] {0.98};
%\addplot[domain=0:30000,help lines] {0.97};
%\addplot[domain=0:30000,help lines] {0.96};
%\addplot[domain=0:30000,help lines] {0.95};
\end{axis}
\end{tikzpicture}
\par\end{center}
\par\end{centering}

\caption{\label{fig:nodes-with-one-cache-entry}Percentage of nodes with only
one memoization entry for \protect\code{derive}}
\end{figure}

The implementation in \citet{Might:2011:PDF:2034773.2034801} uses
hash tables to memoize \code{nullable?}, \code{derive}, and \code{parse-null}.
Function arguments are looked up in those hash tables to see if a
result has already been computed and, if so, what that result is.
Unfortunately, hash tables can be slow relative to other operations.
For example, in simple micro-benchmarks we found that that Racket's
implementation of hash tables can be up to 30 times  slower than
field access. Since memoization-table access is so central to the
memoization process, we want to avoid this overhead. We do so by storing
memoized results as fields in the nodes for which they apply instead
of in hash tables mapping nodes to memoized results.

This technique works for \code{nullable?} and \code{parse-null},
but \code{derive} has a complication. The \code{derive} function
is memoized over not only the input grammar node but also the token
by which that node is being derived. Thus, for each grammar node,
there may be multiple memoized results for multiple different tokens.
The implementation used by \citet{Might:2011:PDF:2034773.2034801}
handles this using nested hash tables. The outer hash table maps grammar
nodes to inner hash tables that then map tokens to the memoized result
of \code{derive}. While we can eliminate the outer hash table by
storing the inner hash tables for \code{derive} in a field in individual
grammar nodes, the central importance of \code{derive} makes eliminating
both sorts of hash table desirable.

These inner hash tables are usually small and often contain only a
single entry. \prettyref{fig:nodes-with-one-cache-entry} shows the
percentage of inner hash tables in the original PWD implementation
that have only a single entry when parsing files from the Python Standard
Library. Though we note the grouping into two populations, what interests
us is that so many have only one entry. We can optimize for the single-entry
case by adding two fields to each grammar node that behave like the
key and value of a hash table that can store only one entry, and when
a second entry is added, evicts the old entry.

This makes our memoization forgetful, and it may fail to notice when
a token is reused multiple times in the input. However, the complexity
results in \prettyref{sec:Complexity-Analysis} still hold, as they
already assume every token is unique. Cycles in the grammar still
require that we not forget the memoizations of \code{derive} on the
current input token, but that requires only the single entry we store
in each node.

\begin{figure}[tb]
\noindent \begin{centering}
\noindent \begin{center}
\begin{tikzpicture}
\begin{axis}[enlarge x limits=false, xmax=27500, ymin=0, ymax=1.6, scaled ticks=false,
 xlabel={Number of tokens in the input},
 ylabel={Uncached calls to \code{derive}},
 every tick/.style=black, minor x tick num=1,
 ytick={0.0,0.2,0.4,0.6,0.8,1.0,1.2,1.4,1.6}, minor y tick num=1,
 yticklabel={\pgfmathparse{\tick*100}\pgfmathprintnumber{\pgfmathresult}\%}]
\addplot[only marks, mark size=0.5pt] table [
 x expr=\thisrowno{0},
 y expr=(\thisrowno{4}+\thisrowno{6})/\thisrowno{9}] {data/uncached2.data};
\addplot[domain=0:30000,help lines] {1};
\end{axis}
\end{tikzpicture}
\par\end{center}
\par\end{centering}

\caption{\label{fig:uncached-calls-to-derive}Percentage of uncached calls
to \protect\code{derive} with single entry versus full hash table}
\end{figure}
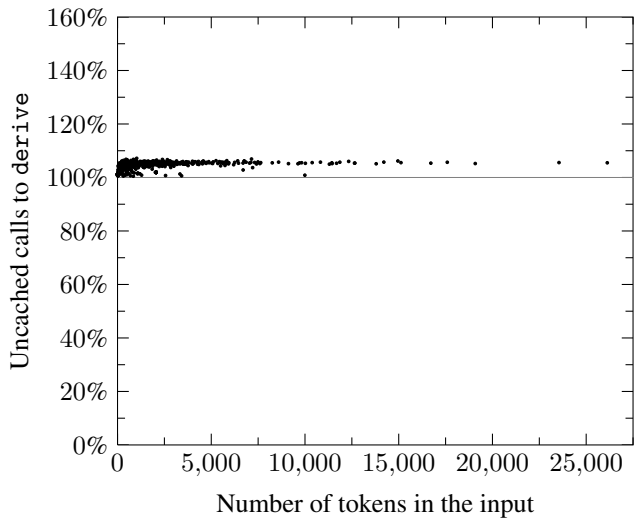

We discovered that, in practice, the number of extra calls to \code{derive}
that are recomputed as a result of this is relatively small. \prettyref{fig:uncached-calls-to-derive}
shows the number of calls to \code{derive} in our implementation
when using the single-entry technique relative to the number when
using full hash tables. While there are more uncached calls when using
the single-entry technique, the increase is on average only 4.2\%
and never more than 4.8\%. We also experimented with larger caches
(e.g., double- or triple-entry caches) to see if the extra cache hits
outweighed the extra computation cost. Early results were not promising,
however, so we abandoned them in favor of a single-entry cache.

\begin{figure}[tb]
\noindent \begin{centering}
\noindent \begin{center}
\begin{tikzpicture}
\begin{axis}[scaled ticks=false, enlarge x limits=false, xmax=27500, %ymode=log,
 every tick/.style=black, minor x tick num=1, minor y tick num=1, ytick={-10,...,10}, %ytickten={-20,...,20},
 ymin=0, ymax=7,
 xlabel={Number of tokens in input},
 ylabel={Speedup factor},
 legend cell align=left,
 legend columns=4, legend pos={south east}]
%\addplot[only marks, mark size=0.5pt] table [
% x expr=\thisrowno{10},
% y expr=\thisrowno{2}/\thisrowno{10}] {data/derp3-times.data};
%\addplot[only marks, mark size=1.5pt, mark=+] table [
% x expr=\thisrowno{22},
% y expr=\thisrowno{3}/\thisrowno{22}] {data/derp3-times-hash.data};
\addplot[only marks, mark size=0.5pt] table [
 x expr=\thisrowno{22},
 y expr=\thisrowno{14}/\thisrowno{3}] {data/derp3-times2-hash-rel.data};
\addplot[domain=0:30000,help lines] {1};
\end{axis}
\end{tikzpicture}
\par\end{center}
\par\end{centering}

\caption{\label{fig:single-entry-speed-up}Performance speedup of single entry
over full hash tables}
\end{figure}
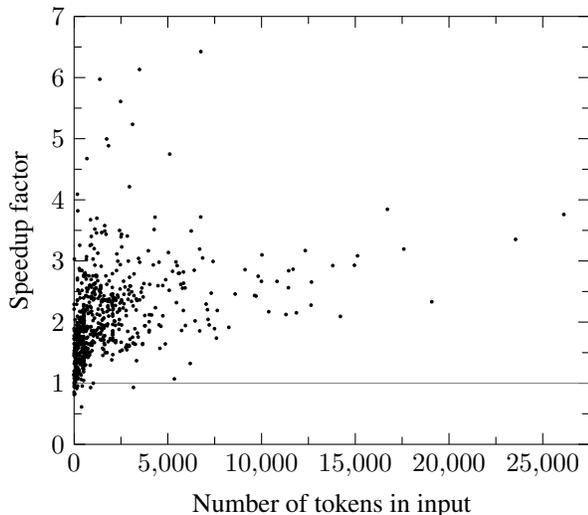

We measured the performance impact of this by running our implementation
both with the single-entry technique and with full hash tables. The
relative speedup of using the single-entry technique is shown in \prettyref{fig:single-entry-speed-up}.
The extra calls to \code{derive} partially cancel out the performance
improvements from avoiding the inner hash tables, but on average the
performance still speeds up by a factor of 2.04.

\section{Conclusion}

In this paper, we have discovered that the believed poor performance
of PWD both in theory and practice is not inherent to PWD. Rather,
its worst-case performance at $O\mathopen{}\left(n^{3}\right)\mathclose{}$
is comparable to other full CFG parsers. Furthermore, with only a
few algorithmic tweaks, the unacceptably slow performance of the implementation
in \citet{Might:2011:PDF:2034773.2034801} can be sped up by a factor
of 951 to be on par with other parsing frameworks.

\acks{This material is partially based on research sponsored by DARPA under
agreements number AFRL FA8750-15-2-0092 and FA8750-12-2-0106 and by
NSF under CAREER grant 1350344. The U.S. Government is authorized
to reproduce and distribute reprints for Governmental purposes notwithstanding
any copyright notation thereon.\vspace{0.15em}
}

\appendix

\balance\bibliographystyle{plainnat}
\phantomsection\addcontentsline{toc}{section}{\refname}\bibliography{derivatives2}

\end{document}

%% file: util/lyxproofs.tex
\global\long\def\proofsep{\mathbin{\vdash}}

\global\long\def\fCenter{\proofsep}

\global\long\def\binaryprimitive#1#2{\BinaryInf$#1\fCenter#2$}

\global\long\def\axiomprimitive#1#2{\Axiom$#1\fCenter#2$}

\global\long\def\unaryprimitive#1#2{\UnaryInf$#1\fCenter#2$}

\global\long\def\trinaryprimitive#1#2{\TrinaryInf$#1\fCenter#2$}

\global\long\def\bussproof#1{#1\DisplayProof}

\global\long\def\binaryinfc#1#2#3#4{#1#2\RightLabel{\ensuremath{{\scriptstyle \textrm{#4}}}}\BinaryInfC{\ensuremath{#3}}}

\global\long\def\trinaryinfc#1#2#3#4#5{#1#2#3\RightLabel{\ensuremath{{\scriptstyle \textrm{#5}}}}\TrinaryInfC{\ensuremath{#4}}}

\global\long\def\unaryinfc#1#2#3{#1\RightLabel{\ensuremath{{\scriptstyle \textrm{#3}}}}\UnaryInfC{\ensuremath{#2}}}

\global\long\def\axiomc#1{\AxiomC{\ensuremath{#1}}}

\global\long\def\binaryinf#1#2#3#4#5{#1#2\RightLabel{\ensuremath{{\scriptstyle \textrm{#5}}}}\binaryprimitive{#3}{#4}}

\global\long\def\trinaryinf#1#2#3#4#5#6{#1#2#3\RightLabel{\ensuremath{{\scriptstyle \textrm{#6}}}}\trinaryprimitive{#4}{#5}}

\global\long\def\unaryinf#1#2#3#4{#1\RightLabel{\ensuremath{{\scriptstyle \textrm{#4}}}}\unaryprimitive{#2}{#3}}

\global\long\def\axiom#1#2{\axiomprimitive{#1}{#2}}

\global\long\def\axrulesp#1#2#3#4{\bussproof{\unaryinf{\axiomc{\vphantom{#4}}}{#1}{#2}{#3}}}

\global\long\def\axrule#1#2#3{\axrulesp{#1}{#2}{#3}{\proofsep\Gamma,\Delta}}

\global\long\def\unrule#1#2#3#4#5{\bussproof{\unaryinf{\axiom{#1}{#2}}{#3}{#4}{#5}}}

\global\long\def\unrulec#1#2#3{\bussproof{\unaryinfc{\axiomc{#1}}{#2}{#3}}}

\global\long\def\binrule#1#2#3#4#5#6#7{\bussproof{\binaryinf{\axiom{#1}{#2}}{\axiom{#3}{#4}}{#5}{#6}{#7}}}